%% file: main.tex
\RequirePackage{amsmath}
\documentclass{llncs}

\input{preamble.tex}

\begin{document}
	
	\title{More Effort Towards Multiagent Knapsack}
	\author{Sushmita Gupta\inst{1} \and
		Pallavi Jain\inst{2} \and
		Sanjay Seetharaman\inst{1}}
	\institute{The Institute of Mathematical Sciences, HBNI, Chennai, India \\ \email{\{sushmitagupta,sanjays\}@imsc.res.in} \and
		IIT Jodhpur, Jodhpur, India \\ \email{pallavi@iitj.ac.in}}

	\maketitle
	
	\begin{abstract}
		In this paper, we study some multiagent variants of the knapsack problem. Fluschnik et al. [AAAI 2019] considered the model in which every agent assigns some utility to every item. They studied three preference aggregation rules for finding a subset (knapsack) of items: individually best, diverse, and Nash-welfare-based. Informally, diversity is achieved by satisfying as many voters as possible. Motivated by the application of aggregation operators in multiwinner elections, we extend the study from diverse aggregation rule to Median and Best scoring functions. We study the computational and parameterized complexity of the problem with respect to some natural parameters, namely, the number of voters, the number of items, and the distance from an easy instance. We also study the complexity of the problem under domain restrictions. Furthermore, we present significantly faster parameterized algorithms with respect to the number of voters for the diverse aggregation rule.
	\end{abstract}

\input{intro.tex}

	\input{prelim.tex}

	\input{hardness.tex}

\input{special_cases.tex}

\input{conclusion.tex}

	\bibliographystyle{splncs04}
	\bibliography{references-AAMAS'22.bib, Voting.bib}
	\clearpage
	\appendix

\input{appendix.tex}
	
\end{document}

%% file: preamble.tex
\PassOptionsToPackage{usenames,dvipsnames}{xcolor}

\usepackage{amsmath}

\usepackage{amssymb}
\usepackage{amsthm}
\usepackage{epsfig}
\usepackage{pgf,tikz,pgfplots}
\usepackage{mathrsfs}
\usetikzlibrary{arrows.meta}
\tikzset{%
	tipA/.tip={Triangle[angle=45:9pt]},
	tipB/.tip={Triangle[angle=45:9pt]}
}
\usepackage{color}
\usepackage{todonotes}
\usepackage{soul}
\usepackage{csquotes}
\usepackage{multirow}
\usepackage{tabularx}
\usepackage{booktabs}
\usepackage{array}
\usepackage{nicefrac}
\usepackage{tcolorbox}
\usepackage{nicefrac}
\usepackage{xspace}

\usepackage{comment}
\excludecomment{suppress}

\usepackage{enumerate}
\usepackage[inline]{enumitem}

\usepackage{hyperref}
\usepackage{cleveref}

\usepackage{color}

\newcolumntype{C}[1]{>{\centering\let\newline\\\arraybackslash}m{#1}}
\newcommand{\OO}{{\mathcal O}}

\newcommand{\nph}{\textsf{NP}-hard\xspace}

\newcommand{\W}{\textsf{W}\xspace}
\newcommand{\XP}{\textsf{XP}\xspace}
\newcommand{\fptas}{\textsf{FPTAS}\xspace}

\newcommand{\fpt}{{\sf FPT}\xspace}

\newcommand{\X}[1]{\ensuremath{\mathscr{#1}}\xspace}

\crefname{theorem}{Thm.}{Thm.}
\Crefname{theorem}{Theorem}{Theorems}
\crefname{lemma}{Lem.}{Lem.}
\Crefname{lemma}{Lemma}{Lemmas}
\crefname{corollary}{Cor.}{Cor.}
\Crefname{corollary}{Corollary}{Corollaries}
\DeclareMathOperator*{\argmax}{arg\,max}
\usepackage[linesnumbered, ruled, algo2e]{algorithm2e}

\newcommand\dk{\textsc{Diverse Knapsack}\xspace}

\usepackage{multirow}

\newcommand{\util}{\ensuremath{\mathtt{util}\xspace}}
\newcommand{\Wh}[1]{\textsf{W}[#1]-hard\xspace}
\newcommand{\budget}{\ensuremath{\mathtt{b}}\xspace}
\newcommand{\T}[1]{\ensuremath{\mathtt{#1}}\xspace}

\newcommand{\median}[2]{\ensuremath{ \mathtt{sat}_{#1}^{#2}}}
\newcommand{\best}[2]{\ensuremath{ \mathtt{sat}_{#1}^{[#2]}}}

\newcommand\mpb{\textsc{Median-UK}\xspace}
\newcommand\muk{\textsc{Median-UK}\xspace}
\newcommand\bpb{\textsc{Best-UK}\xspace}
\newcommand\buk{\textsc{Best-UK}\xspace}

\newcommand\mpbfull{\textsc{Median-Utilitarian Knapsack}\xspace}

\newcommand\ccfull{\textsc{Chamberlin Courant}\xspace}
\newcommand\cc{\textsc{CC}\xspace}
\newcommand\setcover{\textsc{Set Cover}\xspace}
\newcommand\kpcover{\textsc{Knapsack Cover}\xspace}
\newcommand\kp{\textsc{Knapsack}\xspace}
\newcommand\wmscfull{\textsc{Weighted Set Multicover}\xspace}
\newcommand\wmsc{\textsc{WSMc}\xspace}

\newtheorem{rr}{Reduction Rule}

\newcommand{\defproblem}[3]{
	\vspace{1mm}
	\noindent\fbox{
		\begin{minipage}{0.95508\textwidth}
			\begin{tabular*}{\textwidth}{@{\extracolsep{\fill}}lr} #1 \\ \end{tabular*}
			{\bf{Input:}} #2  \\
			{\bf{Question:}} #3
		\end{minipage}
	}
	\vspace{1mm}
}

\newcommand{\ma}[1]{\todo[color={SpringGreen}]{\small{#1}}}

\newcommand{\hide}[1]{}

%% file: intro.tex
	\section{Introduction}\label{intro}
	
	Knapsack is a paradigmatic problem in the area of optimization research, and its versatility in modeling situations with dual objectives/criterion is well established~\cite{KnapsackB}. Unsurprisingly, it has been generalized and extended to incorporate additional constraints; and encoding {\it preferences} is a move in that direction. This type of modeling allows us to address the need for mechanisms to facilitate complex decision-making processes involving multiple agents with competing objectives while dealing with limited resources. We will use {\it multiagent knapsack} to refer to this setting. One natural application of multiagent knapsack would be in the area of participatory budgeting (PB, in short), a democratic process in which city residents decide on how to spend the municipal budget.

	In the context of multiagent knapsack, approval ballot/voting is among the most natural models. It fits quite naturally with the motivation behind PB, in which for a given set of projects, every voter approves a set of projects that s/he would like to be executed. However, {\it utilitarian voting} (variously called score or range voting) is a more enriched ballot model, where every voter expresses his/her preferences via a utility function that assigns a numerical value to each alternative on the ballot (approval ballot is a special case as $0$ utility denotes disapproval). This model applies to PB very well in situations where residents do not really have the motivation to reject any project.  For example, consider the proposal to build different sporting facilities in a city. One might like a sport over another but would not really have any objections against building facilities for {\it any} sport. So, instead of disapproving a project, residents could be asked to give some numerical value (called utility) to every project, depending on how much they value that project. In a realistic scenario, citizens would be asked to ``rate'' the proposed projects by a number between 1 to 10, with 10 being the highest. A larger (smaller) range can be conveniently chosen if the number of proposed projects is reasonably high (low). %

	Recently, this model has been studied by Fluschnik et al.~\cite{fluschnik2019fair}, analysing the problem from a computational viewpoint; and  Aziz and Lee~\cite{DBLP:conf/aaai/0001L21} who considered the axiomatic properties. Formally, we define the problem as follows:

	\defproblem{${\cal R}$-{\sc Utilitarian Knapsack} (${\cal R}$-{\sc UK})}{A set of items $\X{P}=\{p_1,\ldots,p_m\}$, a cost function $\mathtt{c}: \X{P} \rightarrow \mathbb{N}$, a set of voters $\X{V}$, a utility function $\mathtt{util}_v: \X{P} \rightarrow \mathbb{N}$ for every voter $v\in \X{V}$, a budget $\budget \in \mathbb{N}$, and a target $\mathtt{u_{tgt}} \in \mathbb{N}$.}{Does there exist a set $\X{Z}\subseteq \X{P}$ such that the cost of $\X{Z}$ is at most $\mathtt{b}$ and the total {\it satisfaction} (defined below) of voters is at least $\T{u_{tgt}}$ under the voting rule ${\cal R}$?}

	Notably, {\sc Committee Selection} problem is a special case of PB (or UK) where the \textit{candidates} (items) are of unit cost, and we are looking for a \textit{committee} (bundle) of size exactly $\mathtt{b}$. An important and well-studied class of voting rules for {\sc Committee Selection} is the class of {\em Committee Scoring Rules} (CSR)~\cite{endriss2017trends,elkind2017properties,faliszewski2019committee}. In these voting rules, we define the {\em satisfaction} of a voter as a function that only depends on the utility of committee members, assigned by the voter. Towards this, for a committee $S$, we define a {\em utility vector} of the voter $v$ as a vector of the utilities of the candidates in $S$, assigned by $v$. For example, let $\X{P}=\{p_1,p_2,p_3\}$. Consider a voter $v$. Let $\mathtt{util}_v(p_1)=3$, $\mathtt{util}_v(p_2)=2$, and $\mathtt{util}_v(p_3)=5$. Then, for $S=\{p_1,p_2\}$, the utility vector of $v$ is $(3,2)$. The scoring function, $f$, takes the utility vector and returns a natural number as the satisfaction of a voter. Given a committee scoring function, various voting rules can be defined. The most commonly studied one aims to find a committee that maximizes the summation of the satisfaction of all the voters.

	\noindent
	\textbf{Background.} Fluschnik et al.~\cite{fluschnik2019fair} propose three CSRs for ${\cal R}$-{\sc PB}, and one of the rules uses a scoring function that takes the utility vector of a voter and returns the maximum utility in that vector as the satisfaction of the voter. This rule is a generalization of Chamberlin–Courant multiwinner (CC, in short, also known as the 1-median) voting rule, where the utilities are given by the Borda scores. This is called the \dk problem. In this paper, we consider the {\em median scoring function} (and the {\em best scoring function})~\cite{skowron2016finding}, in which given a value $\lambda \in \mathbb{N}$, the scoring function takes a utility vector and returns the $\lambda^{\text{th}}$ maximum value (the sum of the top $\lambda$ values) as the satisfaction of the voter. If the size of the bundle is less than $\lambda$, then the utility of every voter under the median scoring function is $0$. For a bundle $\X{Z}\subseteq \X{P}$ and $\lambda \in \mathbb{N}$, let $\mathtt{sat}_v^{\lambda}(\X{Z})$ ($  \mathtt{sat}_v^{[\lambda]}(\X{Z})$) denote the satisfaction of the voter $v$ from the set \X{Z} under median (best) scoring function. These rules are generalizations of the  $1$-median and $k$-best rules that are widely found in the literature. To the best of our knowledge, none of these are actually applied to real-life situations to determine the ``winner''. As with almost any mathematical model dealing with social choice, these are also proposals with provable guarantees and may be applied in practice. Let us consider a hypothetical PB situation where the proposed projects are building schools in different parts of the city, and the citizens are invited to vote for possible locations of choice. It is highly unlikely that after the facilities are built, all the students are admitted to their top choice but are guaranteed to be admitted to one of their top three choices. Then, aggregating preferences based on maximizing the utility of the $\lambda^{th}$ most-preferred school for $\lambda = 3$ takes that uncertainty into account;\hide{(For example, in places where all students are guaranteed a place within their top three choices, $\lambda$ could be set to 3.) At the very least, $\lambda$-median rule} and gives a lower bound on the utility derived by all who get admitted to one of their top $\lambda$ choices.
	
	Formally stated, we consider the problems \mpbfull (\mpb, in short) and {\sc Best-Utilitarian Knapsack} (\bpb, in short) where the inputs are the same, and the goal is to decide if there exists a bundle $\X{Z}\subseteq \X{P}$ such that $\sum_{p\in \X{Z}}\T{c}(p) \leq \budget$ and $\sum_{v\in \X{V}}\median{v}{\lambda}(\X{Z})\geq \mathtt{u_{tgt}}$ for the former; and $\sum_{v\in \X{V}}\best{v}{\lambda}(\X{Z})\geq \mathtt{u_{tgt}}$ for the latter.  In the optimization version of this problem, we maximize the satisfaction of the voters. We use both variants (decision and optimization) in our algorithmic presentation, and the distinction will be clear from the context. Without loss of generality, we may assume that the cost of each item is at most the budget \budget. To be consistent with the literature, when $\lambda=1$, we refer to both \muk and \buk as \dk.\footnote{Missing details/proofs are in the Appendix.}

	\noindent\textbf{A generalized model.} It is worthwhile to point out that our work and that of Fluschnik et al.~\cite{fluschnik2019fair} are initial attempts at studying the vast array of problems in the context of multiagent knapsack with various preference elicitation schemes and various voting rules - this includes ${\cal R}$-{\sc UK} as a subclass. For over a decade, researchers have studied scenarios where the votes/preferences are consistent with utility functions \cite{ProcacciaRosenschein06,CaragiannisProcaccia11,Boutilier15,AnshelevichBhardwajPostl15,AnshelevichPostl17j,AnshelevichSekar16,ABEPS18,BPNS17}. \hide{Our choice is utility functions with ordered weighted average (OWA) operators.} \hide{One can define a myriad of problems by varying the preference elicitation method and the voting rule; and this would include ${\cal R}$-{\sc UK} as a subclass.} \hide{Each would be interesting and worthwhile to study in its own way. This type of generalization is not a new concept in the area of social welfare maximization under a budget constraint using various voting rules with or without utility functions. For over a decade, researchers have studied scenarios where the votes/preferences are consistent with utility functions or that they can be interpreted as utility functions. \ma{shorten. remove names of authors} Procaccia and Rosenschein \cite{ProcacciaRosenschein06}, Caragiannis and Procaccia~\cite{CaragiannisProcaccia11}, Boutilier et. al~\cite{Boutilier15}, Anshelevich and Bhardwaj et. al~\cite{AnshelevichBhardwajPostl15}, Anshelevich and Sekar~\cite{AnshelevichSekar16}, Anshelevich and Postl~\cite{AnshelevichPostl17j}, and Anshelevich and Bhardwaj et. al~\cite{ABEPS18}, Benad\'e et. al~\cite{BPNS17} to name a few.}

	\noindent \textbf{Our contributions.} In this paper, we study the computational and parameterized complexity of \mpb with respect to various natural input parameters and have put more effort towards identifying tractable special cases amid a sea of intractability. Moreover, we extend those results to \bpb.

	For a start, we show that both the problems are \nph for every $\lambda$ (\Cref{thm:im-sc-np-hardness}). Since the {\sc Committee Selection} problem with median scoring function is \Wh{1} with respect to $n$, the number of voters, for all $\lambda>1$~\cite{bredereck2020parameterized}, it follows that its generalization \muk must be as well (\Cref{cor:muk-w-hardness}). Hence, even though an \fpt\footnote{\fpt with respect to parameter $k$ means that there exists an algorithm that solves the problem in $f(k).{\rm poly}(n,m)$ time} algorithm with respect to $n$ is unlikely for \muk, we are able to present an \XP algorithm with respect to $n$ (\Cref{thm:muk-xp}), that runs in time ${\cal O}((m(\lambda+1))^n \text{poly}(n,m))$. Additionally, it is known that \muk with $\lambda=1$ (i.e. \dk), binary utilities, and unary costs is \Wh{2} with respect to budget \budget \cite{fluschnik2019fair}. In the case of parameterization by both $n$ and $\mathtt{b}$, we obtain an \fpt algorithm for both the problems. While there is a trivial $\OO^{\star}(2^{m})$ algorithm, where $m$ denotes the number of items, for both the problems, unless ETH fails there cannot be a $\OO(2^{o(m+n)}{\rm poly}(n,m))$ algorithm \cite{fluschnik2019fair}. In light of these dead ends, we turn our attention to identifying tractable special cases, and our search forked into three primary directions: the value of $\lambda$, preference profiles, and encoding of utilities and costs. When $\lambda=1$, the problem is \nph \cite{fluschnik2019fair} and remains so even when profile is single-crossing or single-peaked. We show that despite this intractability, when the profile is {\it unanimous}, that is, all the voters have the exact same top preference, \dk is polynomial-time solvable (\Cref{lem:dk-unanimous}); but both \muk and \buk remain \nph for unanimous as well as single-crossing profiles (\Cref{thm:im-sc-np-hardness}) when $\lambda>1$. In the situation where there is a priority ordering over all the items, captured by a {\it strongly unanimous} profile, \muk is polynomial-time solvable. Furthermore, we observe that given an instance, the ``closer'' it is to a strongly unanimous profile, the faster the solution can be computed. Consequently, we consider $d$, the distance away from strong unanimity, as a parameter and show that \dk is \fpt with respect to $d$. In the special case where $\lambda=1$ and the utilities or costs are polynomially bounded or are encoded in unary, we have three different \fpt algorithms, all of which are improvements on the one given by~ \cite{fluschnik2019fair}. Additionally, when $\lambda=1$ and the profile is single-peaked (or crossing) the problem admits an \fptas (\Cref{thm:dk-fptas}), an improvement over the $(1-1/e)$-factor approximation algorithm by \cite{fluschnik2019fair}. Table~\ref{table:results} summarizes the main results with precise running times of our algorithms.

	\begin{table}
		\centering
		\begin{tabular}{| C{4cm} | C{5cm} | C{2.5cm} | }
			\toprule
			\textbf{Restriction}                              & \textbf{Result}                                             & \textbf{Ref.}                                                                         \\
			\midrule
			\small{SC}                                        & \multirow{2}{4em}{\small{\nph}}                             & \multirow{2}{4em}{\small{\cref{thm:im-sc-np-hardness}}} \\
			\small{$\lambda>1$ \& U}                          &                                                             &                                                                                       \\
			\midrule
			& \small{$\mathcal{O}(n!\ \text{poly}(\hat{u},n,m))$}         & \cite{fluschnik2019fair}                                                              \\
			& \small{$\mathcal{O}(4^n \text{poly}(\bar{u},n,m))$}         & \small{\cref{thm:dk-4n}}                                                              \\
			\small{$\lambda=1$}                               & \small{$\mathcal{O}(4^n \text{poly}(\mathtt{b},n,m))$}      & \cref{cor:dk-4n}                                                                      \\
			& \small{${\cal O}(2^n \text{poly}(\bar{u},\mathtt{b},n,m))$} & \small{\cref{thm:dk-2n}}                                                              \\
			\midrule
			\multirow{2}{10em}{\small{$\lambda > 1$ \& \muk}} & \small{${\cal O}((m(\lambda+1))^n \text{poly}(n,m))$}       & \small{\cref{thm:muk-xp}}                                                             \\
			& \small{\W[1]-hard w.r.t. $n$}                               & \small{\cite{bredereck2020parameterized}}                                             \\
			\midrule
			No restriction                                    & \small{$\mathcal{O}^*(b^2 (b!)^n 2^{b})$}                   & \small{\cref{thm:muk-fpt-nb}}                                  \\
			\midrule
			\small{SU \& \muk}                                & \small{${\rm poly}(n,m)$}                                   & \small{\cref{lem:muk-strongly-unan}}                                                  \\
			\small{SU \& \buk}                                & \small{${\rm poly}(\bar{u},n,m)$, ${\rm poly}(\T{b},n,m)$}                           & \small{\cref{cor:buk-strongly-unan}}                                                  \\
			\small{\muk}                                      & \small{\XP w.r.t. $d$, \fpt w.r.t. $\mathtt{b},d$}          & \small{\cref{cor:muk-par-d}}                                \\
			\small{\buk}                                      & \small{\fpt w.r.t. $\mathtt{b},d$}                          & \small{\cref{cor:muk-buk-par-d}}                                                      \\
			\midrule
			\small{$\lambda=1$}                               & \small{\fpt w.r.t. $d$ for unary utilities}                 & \small{\cref{cor:dk-par-d}}                                                           \\
			\small{$\lambda=1$ \& SP/SC}                      & \small{\fptas}                                              & \small{\cref{thm:dk-fptas}}                                                           \\
			\small{$\lambda=1$ \& U}                          & \small{${\cal O}(m)$}                                       & \small{\cref{lem:dk-unanimous}}                                                       \\
			\bottomrule
		\end{tabular}
		\caption{Results on \muk and \buk. \hide{Abbreviations: SC-single-crossing, SP-single-peaked, U-unanimous, SU-strongly unanimous.}
			Here $n$, $m$, \budget, and $d$ denote the number of voters, the number of items, the budget, and the distance away from SU, respectively. The abbreviations SC, SP, SU, and U are preference restrictions which are defined in \Cref{intro}.}
		\label{table:results}
	\end{table}

%% file: prelim.tex
	\noindent\textbf{Preliminaries.} Let $\X{V} = \{ v_1, \dots, v_n \}$ be a set of $n$ voters and $\X{P} = \{ p_1, \dots, p_m \}$ be a set of $m$ items. The preference ordering of a voter $v \in \X{V}$ over the set of items is given by the utility function $\mathtt{util}_v : \X{P} \rightarrow \mathbb{N}$. That is, if $\mathtt{util}_v(p) > \mathtt{util}_v(p')$, then $v$ prefers $p$ more than $p'$, and we use the notation $p \succ_v p'$. We drop the subscript when it is clear from the context. For a subset $Y \subseteq \X{V}$, we use $\mathtt{util}_{Y}(p)$ to denote the utility of the item $p$ to voters in $Y$: $\mathtt{util}_{Y}(p) = \sum_{v \in Y} \mathtt{util}_{v}(p)$. The set of utility functions form the \textit{utility profile}, denoted by $\{\mathtt{util}_{v}\}_{v \in \X{V}}$. For integers $i,j$, we use $[i,j]$ to denote the set $\{i, i+1, \dots, j\}$. For an integer $i$, we use $[i]$ to denote $[1,i]$. \hide{We define $\bar{u} = \sum_{v \in \X{V}} \max_{p \in \X{P}} \mathtt{util}_{v}(p)$ and $\hat{u} = \sum_{p \in \X{P}} \sum_{v \in \X{V}} \mathtt{util}_v(p)$.} In an instance of \muk, we say that an item $p\in \X{P}$ is a \textit{representative} of a voter $v \in \X{V}$ in a bundle $\X{Z}\subseteq \X{P}$ if $p$ is the $\lambda^{th}$ most preferred item of $v$ in $\X{Z}$. %

	\noindent\textbf{Preference profiles.} A preference profile (P, in short) is said to be
	\vspace{-1em}
	\begin{itemize}[wide=0pt,itemsep=0pt]
		\item {\it unanimous} (U) if all voters have the same top preference;
		
		\item {\it strongly unanimous} (SU) if all voters have the identical preference orderings;
		
		\item {\it single-crossing} (SC) if there exists an ordering $\sigma$ on voters \X{V} such that for each pair of items $\{p, p'\} \subseteq \X{P}$, the set of voters $\{v \in \X{V}: \util_{v}(p) \geq \util_{v}(p')\}$ forms a consecutive block according to $\sigma$;
				
		\item {\it single-peaked} (SP) if the following holds for some ordering, denoted by $\triangleleft$, on the items \X{P}: Let $\T{top_{v}}$ denote voter $v$’s most preferred item. Then, for each pair of items $\{p, p'\} \subseteq \X{P}$ and each voter $v \in \X{V}$, such that $p \triangleleft p' \triangleleft \T{top_{v}}$ or $\T{top_{v}} \triangleleft p' \triangleleft p$ we have that $v$ weakly prefers $p'$ over $p$, i.e $\util_{v}(p') \geq \util_{v}(p)$.
	\end{itemize}

%% file: hardness.tex
	\section{Hardness of \muk} \label{pc}
	
	We dedicate this section to identifying intractable cases. Clearly, \muk is \nph as its unweighted version, {\sc Committee Selection}, under the same rule, is \nph \cite{skowron2016finding}. We begin with the following strong intractability result.
	
	\begin{theorem}
		\label{thm:im-sc-np-hardness}\label{cor:buk-np-hardness}\label{nph-sc}
		\muk and \buk are \textup{\nph} for every $\lambda\!\in\!\mathbb{N}$ even for SCPs. Furthermore, for $\lambda\!>\!1$, both are \textup{\nph} even for UPs.
	\end{theorem}
	
	\begin{proof}
		We give a polynomial-time reduction from the \dk problem, which is known to be \nph for SCPs \cite{fluschnik2019fair}. Let ${\cal I}=(\X{P},\X{V},\mathtt{c},\{\mathtt{util}_v\}_{v\in \X{V}},\\\mathtt{b}, \mathtt{u_{tgt}})$ be an instance of \dk with SCP. Let $\sigma = (v_1, \ldots, v_n)$ be an SC ordering of the voters in $\X{V}$. Without loss of generality,  let the preference order of $v_1$ be $p_1 \succ p_2 \succ \ldots \succ p_m$. Let $u_{max}$ denote the maximum utility that a voter assigns to an item.

		To construct an instance of \muk, the plan is to add $\lambda-1$ new items, say $\X{P}_{new} = \{ p_{m+1}, \dots, p_{m+\lambda-1} \}$, to $\X{P}$ and ensure that they are in any feasible bundle by setting appropriate utilities and costs. Formally, we construct an instance ${\cal I}^\prime=(\X{P}^\prime,\X{V},\mathtt{c}^\prime,\{\mathtt{util}^\prime_v\}_{v\in \X{V}},\mathtt{b}^\prime, \mathtt{u_{tgt}}, \lambda)$ of \muk as follows.  Let $\X{P}^\prime = \X{P} \cup \X{P}_{new}$. For $p\in \X{P}$, let $\mathtt{c}'(p)=\mathtt{c}(p)$, and for $p \in \X{P}_{new}$, let $\mathtt{c}'(p) = 1$; and we set the budget $\mathtt{b}^\prime = \mathtt{b}+\lambda-1$. Note that the set of voters and the target utility in the instance ${\cal I}'$ are the same as in the instance ${\cal I}$. Next, we construct the utility function for a voter $v\in \X{V}$ as
		\[
		\mathtt{util}_{v}^\prime(p_j) =
		\begin{cases}
		u_{max} + j - m      & \quad\text{if } p_j \in \X{P}_{new}, \\
		\mathtt{util}_v(p_j) & \quad\text{otherwise.}              \\
		\end{cases}
		\]

		Clearly, the construction is doable in polynomial time. Without loss of generality, we assume that $\mathtt{u_{tgt}} > 0$, otherwise we return an empty bundle.

		The rest of the proof is in \Cref{proof-omit}. In the case of \buk, we can use the same reduction as above. From the construction, it follows that $\X{P}_{new}$ belongs to any feasible bundle. We obtain the hardness result proceeding similarly. 
	\end{proof}

	Since {\sc Committee Selection} problem is  {\sf W}[1]-hard with respect to $n$ for $\lambda>1$~\cite{bredereck2020parameterized}, we have the following result.
	
	\begin{corollary}
		\label{cor:muk-w-hardness}
		\muk is \textup{{\sf W}[1]-hard} with respect to $n$ for $\lambda>1$.
	\end{corollary}
	
	Hence, it is unlikely to be \fpt. Now, we show that there exists an \XP algorithm for \muk with respect to $n$. The intuition is as follows. For every voter $v\in \X{V}$, we guess the representative in the solution. We know that for every voter $v$, we must pick $\lambda-1$ other items whose utility is larger than the representative's. Thus, we reduce the problem to the \wmscfull (\wmsc) problem, which is defined as follows:

	\defproblem{\wmscfull (\wmsc)}{A universe $U$, a family, ${\cal F}$, of subsets of $U$, a cost function $\mathtt{c}\colon {\cal F}\rightarrow \mathbb{N}$, a budget $\mathtt{b} \in \mathbb{N}$, and a positive integer $k\in \mathbb{N}$.}{Find a subset ${\cal F}^\prime \subseteq {\cal F}$ such that every element of $U$ is in at least $k$ sets in ${\cal F}^\prime$, and $\sum_{F\in {\cal F}^\prime}\mathtt{c}(F)\leq \mathtt{b}$.}

	If a set $F$ contains an element $u\in U$, then we say that $F$ {\em covers} $u$. We first present an \fpt algorithm for \wmsc with respect to $k+n$, which is similar to the \fpt algorithm for {\sc Set Cover} with respect to $n$.
	
	\begin{lemma}\label{lem:fpt-wmsc}
		\wmsc can be solved in $\OO((k+1)^n|{\cal F}|)$ time.
	\end{lemma}

	\begin{proof}
		We give a dynamic-programming algorithm. Let $(U,{\cal F}, \mathtt{c}, \mathtt{b},k)$ be an instance of \wmsc. Let $U=\{u_1,\ldots,u_n\}$ and ${\cal F}=\{F_1,\ldots,F_m\}$. For a set $S\subseteq U$, let $\chi(S)$ denote the characteristic vector of $S$, i.e., it is an $|U|$-length vector such that $\chi(S)_{i}=1$ if and only if $u_i\in S$. For any two $|U|$-length vectors $\overrightarrow{A},\overrightarrow{B}$, we define the difference $\overrightarrow{C}=\overrightarrow{A}-\overrightarrow{B}$ as $\overrightarrow{C}_i = \max(0, \overrightarrow{A}_i-\overrightarrow{B}_i)$ for each $i \in [|U|]$. Let ${\cal R}=\{0,1,\ldots,k\}^n$. We define the dynamic-programming table as follows. For every $\overrightarrow{X}\in {\cal R}$ and $j\in [0,m]$, let $T[\overrightarrow{X},j]$ be the minimum cost of a subset of $\{F_1,\ldots,F_j\}$ that covers $u_i$ at least $\overrightarrow{X}_i$ times for every $i \in [n]$. We compute the table entries as follows.
		
		\paragraph{Base Case:} For each $\overrightarrow{X}\in {\cal R}$, we set
		\begin{equation}\label{eq1:dp-wmsc}
		T[\overrightarrow{X},0]= \begin{cases} 0, \text{ if } \overrightarrow{X}=\{0\}^n, \\
		\infty, \text{ otherwise. }
		\end{cases}
		\end{equation}
		
		\paragraph{Recursive Step:}
		For each $\overrightarrow{X}\in {\cal R}$ and $j\in [m]$, we set
		\begin{equation}\label{eq2:dp-wmsc}
		T[\overrightarrow{X},j]= \min\{T[\overrightarrow{X},j-1],T[\overrightarrow{X}-\chi(F_j),j-1]+\mathtt{c}(F_j)\}.
		\end{equation}
		If there exists an $\overrightarrow{X}\in {\cal R}$ such that $\overrightarrow{X}_{i}\geq k$ for each $i\in [n]$, and the value in the entry
		$T[\overrightarrow{X}, m]$ is at most $\budget$, then we return ``YES", else we return ``NO". We give a proof of correctness in \Cref{proof-omit}.
	\end{proof}
	
	Now, we are ready to give an \XP algorithm for \muk.
	
	\begin{theorem}\label{thm:muk-xp}
		\muk can be solved in $\OO((m(\lambda+1))^n \text{poly}(n,m))$ time.
	\end{theorem}

	\begin{proof}
		We begin by guessing the representative of every voter. Let $r_v$ denote the guessed representative of a voter $v$, and $R$ be the set of all guessed representatives. Next, we construct an instance of \wmsc as follows. For every voter $v\in \X{V}$, we add an element $e_v$ to $U$. Corresponding to every item $p\in \X{P}$, we have a set $F_p$ in ${\cal F}$, where $F_p=\{e_v \colon \mathtt{util}_v(p) \ge \mathtt{util}_v(r_v)\}$. Note that $F_p$ contains the elements corresponding to the set of voters who prefer $p$ at least as much as their representative. For every $F_p \in {\cal F}$, $\mathtt{c}(F_p)=\mathtt{c}(p)$. We set $\mathtt{\hat{b}}=\mathtt{b}$ and $k=\lambda$. Let ${\cal F'} \subseteq {\cal F}$ be the solution of the \wmsc instance $(U, {\cal F}, \mathtt{c}, \mathtt{\hat{b}}, k)$. We construct a set $P'\subseteq \X{P}$ that contains items corresponding to the sets in ${\cal F}'$, i.e., $P'=\{p\in \X{P}\colon F_p \in {\cal F}'\}$. Note that the cost of $P'$ is at most $\mathtt{b}$. Also, every element of $U$ is covered at least $\lambda$ times. Therefore, for every voter $v$, $P'$ contains at least $\lambda$ items whose utilities are at least that of $r_v$. Thus, the total utility is at least $\sum_{v\in \X{V}}\mathtt{util}_v(r_v)$. Note that the cost and utility of a bundle can be computed in time polynomial in $n$ and $m$. Since ``guessing" the representatives takes $m^n$ steps, the running time follows due to \Cref{lem:fpt-wmsc}.
	\end{proof}

	We cannot expect a similar result for \buk due to the following reduction. Let $(A = \{a_1, \dots, a_m\}, \T{b}, w, \T{c}', \T{v_{tgt}})$ be an instance of {\sc Knapsack} problem where each item $a_i$ costs $\mathtt{c}'(a_i)$ and has value $w(a_i)$. The goal is to determine if there exists a knapsack with value at least $\T{v_{tgt}}$ and cost at most $\mathtt{b}$. Construct an instance of \buk with exactly one voter $v_1$, set of items $A$, budget $\T{b}$, target utility $\T{v_{tgt}}$, and $\lambda = m$. An \XP algorithm for \buk will give us a polynomial time algorithm for {\sc Knapsack}, implying \textsf{P=NP}.

	Due to Bredereck et al. \cite{bredereck2020parameterized}, \buk is {\sf{FPT}}($n,\lambda$) when the item costs are equal to $1$. We extend their algorithm to solve both \muk and \buk.

	\begin{theorem}
		\label{thm:muk-fpt-nb}\label{cor:buk-fpt-nb}
		\muk and \buk can be solved in $\mathcal{O}^*(b^2 (b!)^n 2^{b})$ time.
	\end{theorem}
	
	\begin{proof}
		
		Since the item costs are natural numbers, the maximum cardinality of a solution bundle is $\min(\mathtt{b}, m)$. Let $S=\{s_1, \dots, s_k\}$ be a solution bundle. First, we guess the value of $k = |S|$. Next, for each voter $v_i$, we guess the permutation $\psi_i$ of $[k]$: $\psi_i(j)=l$ if and only if $s_j$ is the $l^{th}$ most preferred item of $v_i$ in $S$. Note that $\psi_i$ gives us the projection of $S$ on the preference order of $v_i$: $\succ_i$. Let $\mathbb{I}[\psi_i(j)]$ indicate whether $\psi_i(j) = \lambda$.
		We consider the following complete bipartite graph $G$ with bipartitions $[k]$ and $\X{P}$: the weight of an edge $(l,p)$ is given by $w(l,p) = \sum_{v_i \in \X{V}} \mathbb{I}[\psi_i(l)] \cdot \mathtt{util}_{v_i}(p)$. An edge $(l,p)$ captures item $p$ taking the role of $s_l$ in $S$: $w(l,p)$ is the contribution of $s_l$ to the utility of $S$. Our goal is to find a matching $M$ of size $k$ such that the cost of the bundle corresponding to $M$ is at most $\mathtt{b}$, and the weight of $M$ is maximum. Each of the $m$ items in $\X{P}$ can be encoded using $ln = \lfloor \log_2 m \rfloor +1$ bits. Thus, any matching of size $k$ in $G$ can be encoded using $k \cdot ln$ bits: the first $ln$ bits represent the item that $s_1$ is matched to, and so on. To find $M$, we iterate through all $2^{k \cdot ln}$ binary strings of length $k \cdot ln$ and look for a desired matching.

		Let $M \subset E(G)$ be a matching with the maximum weight over all possible guesses of $k$ and $\{\psi_i : v_i \in \X{V}\}$. We return ``YES" if the weight of $M$ is at least $\mathtt{u_{tgt}}$; otherwise we return ``NO". We give a proof of correctness in \Cref{proof-omit}.

		\textbf{Complexity:} There are $\mathtt{b}$ guesses for the cardinality of $S$, and $\sum_{k \in [\mathtt{b}]}(k!)^n$ guesses for the permutations $\{\psi_i : v_i \in \X{V}\}$. For each guess $k$, $\{\psi_i : v_i \in \X{V}\}$, we iterate through $2^{k \cdot ln}$ strings of length ${k \cdot ln}$, and for each string we spend polynomial time to process. Thus, the total running time is $\mathcal{O}^*(b \cdot b (b!)^n \cdot 2^{b})$.
		
		By redefining the indicator function $\mathbb{I}$ as $\mathbb{I}(\psi_i(j))=1$ if and only if $\psi_i(j) \le \lambda$ and continuing as before, we obtain the result for \buk.
	\end{proof}

%% file: special_cases.tex
	\section{Algorithms for Special Cases} \label{special_cases}
	
	We now present algorithms for some restrictions in the input. First, we consider the case of a SUP in which all voters have identical preference orders.  We first apply the following reduction rule that simplifies the input.
	
	\begin{rr}[$\dagger$]\label{rr:muk-su}	
		If there exist two voters $v_1, v_2 \in \X{V}$ such that their preference orders are identical, then we can ``merge" the two voters:
		Set $\X{V}^\prime = (\X{V} \setminus \{v_1,v_2\}) \cup \{v_{12}\}$ with $\mathtt{util}^\prime_{v_{12}} (p)= \mathtt{util}_{v_1}(p) + \mathtt{util}_{v_2}(p)$, for all $p\in \X{P}$, and for all other voters $\mathtt{util}^\prime_v(p)=\mathtt{util}_v(p)$. The new instance is ${\cal I}' = (\X{P},\X{V}^\prime,\mathtt{c},\{\mathtt{util}'_v\}_{v\in \X{V}},\mathtt{b}, \mathtt{u_{tgt}},\lambda)$
	\end{rr}

	After the exhaustive application of \Cref{rr:muk-su}, the profile contains only one voter. In the case of \buk with one voter, it is equivalent to a knapsack problem on the items. Thus, we have the following lemma.
	\begin{lemma}[$\dagger$]
		\label{lem:muk-strongly-unan}\label{cor:buk-strongly-unan}
		\muk can be solved in polynomial time for SUPs. Furthermore, \buk can be solved in polynomial time for SUPs, when either the budget, or the utilities are encoded in unary.
	\end{lemma}

	Now, we consider the distance from an easy instance as a parameter, which is a natural parameter in parameterized complexity. Let $d$ be the minimum number of voters that need to change their utility function so that the profile is SU. Given an instance of \muk, $d$ can be computed in polynomial time as follows. First, we sort the utility profile. Let $l$ be the length of the longest block of voters with the same utility function. Then, $d=n-l$. We consider the problem parameterized by $d$. This parameter has also been studied earlier for the {\sc Connected CC} problem~\cite{gupta2020well}.
	
	Without loss of generality, we assume that voters $\{ v_{d+1}, \ldots, v_n \}$ are SU. We first apply \Cref{rr:muk-su} exhaustively. Note that after the exhaustive application of the rule, the instance has at most $d+1$ voters. Since \dk is \fpt with respect to $n$ when the utilities are encoded in unary \cite{fluschnik2019fair}, we have the following corollary.
	\begin{corollary}
		\label{cor:dk-par-d}
		\dk is \fpt with respect to $d$ when utilities are encoded in unary.
	\end{corollary}
	
	Furthermore, due to \Cref{thm:muk-xp,thm:muk-fpt-nb}, we have the following.
	\begin{corollary}
		\label{cor:muk-par-d}
		\label{cor:muk-buk-par-d}
		\hfill
		\begin{enumerate}
			\item[(1)] \muk can be solved in $\OO((m(\lambda+1))^{d+1}\text{poly}(d,m))$ time. 
			\item[(2)] \muk and \buk can be solved in $\mathcal{O}^*(b^2 (b!)^{d+1} 2^{b})$ time.
		\end{enumerate}
	\end{corollary}

	\subsection{When $\lambda = 1$: \dk}
	
	 Unlike \muk for $\lambda>1$, \dk can be solved in polynomial time for UPs. The optimal solution contains the most preferred item of all voters.
	 
	\vspace{-0.4em}
	\begin{lemma}[$\dagger$]\label{lem:dk-unanimous}
		\dk can be solved in polynomial time for UPs.
	\end{lemma}

	Let ${\cal I}=(\X{P},\X{V},\mathtt{c},\{\mathtt{util}_v\}_{v\in \X{V}},\mathtt{b},\mathtt{u_{tgt}})$ be an instance of \dk. For our subsequent discussions, we define \[\bar{u} = \sum_{v \in \X{V}} \max_{p \in \X{P}} \mathtt{util}_{v}(p) \text{, and } \hat{u} = \sum_{p \in \X{P}} \sum_{v \in \X{V}} \mathtt{util}_v(p).\]

	Next, we design \fpt algorithms with respect to $n$ for \dk. Fluschnik et al.~\cite{fluschnik2019fair} gave an algorithm that runs in $\mathcal{O}(n!\ \text{poly}(\hat{u},n,m))$, which is an \fpt algorithm with respect to $n$ when  the total utility is either unary encoded or bounded by $\text{poly}(n,m)$. \hide{We begin with improving the running time of this algorithm.}

	We first give an algorithm that runs in  $\mathcal{O}(4^n\ \text{poly}(\bar{u},n,m))$ time. Further, we give an algorithm that runs in $\mathcal{O}(2^n\ \text{poly}(\bar{u},\mathtt{b},n,m))$ time. It is worth mentioning that the best known algorithm for \ccfull, a special case of \dk, also runs in $\mathcal{O}(2^n\ \text{poly}(n,m))$ time~\cite{DBLP:conf/ijcai/Gupta00T21}.

	\smallskip
	\noindent\textbf{A $\mathcal{O}(4^n\ \text{poly}(\bar{u},n,m))$ algorithm.}
	To design the algorithm, we first reduce the problem to the following variant of the \setcover problem, which we call \kpcover due to its similarity with the \kp and \setcover problems. Then, using the algorithm for \kpcover (Theorem~\ref{thm_dk-kpcover-complexity}), we obtain the desired algorithm for \dk.

	\defproblem{\kpcover}{Two sets of universe $U_1=\{u_1^1,\ldots,u_1^n\}$ and $U_2=\{u_2^1,\ldots,u_2^m\}$, a family of sets $\X{F}=\{\{F,u\}\colon F\subseteq U_1, u\in U_2\}$, a profit function $\mathtt{profit}\colon \X{F} \rightarrow \mathbb{N}$, a cost function $\mathtt{cost}\colon \X{F} \rightarrow \mathbb{N}$, budget $\mathtt{b} \in \mathbb{N}$, and total profit $\mathtt{p}$.}{Does there exist a set $\X{Z}\subseteq \X{F}$ such that (i) for every two sets $\{F,u\}$ and $\{F',u'\}$ ($u$ can be equal to $u'$) in \X{Z}, $F \cap F' =\emptyset$, (ii) $\bigcup_{\{F,u\} \in \X{Z}} F = U_1$, (iii) $\!\sum_{\{F,u\} \in \X{Z}} \mathtt{cost}(\{F,u\}) \! \leq \mathtt{b}$, and (iv) $\sum_{\{F,u\} \in \X{Z}} \mathtt{profit}(\{F,u\}) \geq \mathtt{p}$?}

	We first discuss the intuition behind reducing \dk to \kpcover. Consider a non-empty bundle $S$. For each voter, there exists an item that represents the voter in the bundle. An item can represent more than one voter. For each subset of voters $X$ and each item $p$, we create a set $\{X,p\}$. The goal is to find a family of sets such that the set of voters is disjointly covered by the voter subsets, and the set of items forms a bundle with utility at least $\mathtt{p}$ and cost at most $\mathtt{b}$.
	We first present the reduction from \dk to \kpcover, which formalizes the intuition, in the following lemma.
	\begin{lemma}\label{lem:dk-kpcover}
		\dk can be reduced to \kpcover in $\OO(2^n \\ \text{poly}(n,m))$ time.
	\end{lemma}

	\begin{proof}
		Given an instance ${\cal I}$, we create an instance ${\cal J}=(U_1,U_2,\X{F},\mathtt{cost},\mathtt{profit},\\ \mathtt{b'}, \mathtt{p})$ of \kpcover as follows. We first construct universe sets: $U_1=\X{V}$, and $U_2=\X{P}$. For each $X\subseteq U_1$ and each $r\in U_2$, we add a set $\{X,r\}$ in the set family $\X{F}$. Note that the size of $\X{F}$ is $2^nm$. Next, we define the cost and the profit functions. For every set $\{X,r\}\in \X{F}$, $\mathtt{cost}(\{X,r\})=\mathtt{c}(r)$ and $\mathtt{profit}(\{X,r\})=\mathtt{util}_X(r)$. We set $\mathtt{b'}=\mathtt{b}$ and $\mathtt{p}=\mathtt{u_{tgt}}$. We give a proof of correctness of the reduction in \Cref{proof-omit}. This completes the proof of \Cref{lem:dk-kpcover}.	
	\end{proof}

	Next, we design an algorithm for \kpcover.
	\begin{lemma}[$\dagger$]\label{thm_dk-kpcover-complexity}
		\kpcover can be solved in  $\mathcal{O}(2^{\vert U_1\vert} \vert\X{F}\vert p_{max})$ time, where $p_{max}=\sum_{F\in \X{F}}\mathtt{profit}(F)$.
	\end{lemma}

	\Cref{lem:dk-kpcover,thm_dk-kpcover-complexity} give us an algorithm for \dk. Analogous to \Cref{lem:muk-strongly-unan}, an alternative dynamic-programming approach to solve \kpcover is that instead of finding the minimum cost of a subset with a particular utility, we find the maximum utility of a subset with a particular cost. This results in an algorithm running in time ${\cal O}(2^{\vert U_1\vert} \vert\X{F}\vert \mathtt{b})$. Thus, we have the following results.
	\begin{theorem}\label{thm:dk-4n}
		\dk can be solved in $\mathcal{O}(4^n\ \text{poly}(\bar{u},n,m))$ time.
	\end{theorem}
	\begin{corollary}\label{cor:dk-4n}
		\dk can be solved in $\mathcal{O}(4^n\ \text{poly}(\mathtt{b},n,m))$ time.
	\end{corollary}
	
	\smallskip
	\noindent\textbf{A $\mathcal{O}(2^n\ \text{poly}(\bar{u},\mathtt{b},n,m))$ algorithm.}
	Next, we give an algorithm that improves the exponential dependency on $n$, but it additionally has a dependency on the budget. In particular, we prove the following theorem.
	\begin{theorem}\label{thm:dk-2n}
		\dk can be solved in $\mathcal{O}(2^n\ \text{poly}(\bar{u},\mathtt{b},n,m))$ time.
	\end{theorem}
	
	To prove \Cref{thm:dk-2n}, we use the technique of polynomial multiplication, which has also been recently used to give a $\mathcal{O}(2^n\ \text{poly}(n,m))$-time algorithm for \cc~\cite{DBLP:conf/ijcai/Gupta00T21}. Our algorithm is similar to the one for \cc. Here, we additionally keep track of the budget.

	Note that  there exists a solution $S$ of ${\cal I}$ that induces an $|S|$-sized partition of the voter set, because if an item is not a representative of any voter, then we can delete this item from the solution, and it still remains a solution.
	We use the method of polynomial multiplication to find such a partition that is induced by some solution. We begin with defining some notations and terminologies, which are same as in \cite{DBLP:conf/ijcai/Gupta00T21}.

	Let $U=\{u_1,\ldots,u_n\}$. For a subset $X\subseteq U$, let $\chi(X)$ denote the characteristic vector of the set $X$: an $|U|$-length vector whose $j^{th}$ bit is $1$ if and only if $u_j \in X$. Two binary strings of length $n$ are said to be {\em disjoint} if for each $i \in [n]$, the $i^{th}$ bit of both strings is not same. Let $\mathcal{H}(S)$ denote the {\em Hamming weight} of a binary string $S$: the number of $1$s in $S$. The Hamming weight of a monomial $x^i$, where $i$ is binary vector, is the Hamming weight of $i$. Let $\mathcal{H}_s(P(x))$ denote the {\em Hamming projection} of a polynomial $P(x)$ to a non-negative integer $s$: the sum of all monomials in $P(x)$ with Hamming weight $s$. Let ${\cal R}(P(x))$ denote the {\em representative polynomial} of $P(x)$: if the coefficient of a monomial is non-zero in $P(x)$, then the coefficient of the monomial is one in $\mathcal{R}(P(x))$. Next, we state a basic result following which we prove \Cref{thm:dk-2n}.

	\begin{lemma}\label{lem:distinct-vectors}
		{\rm \cite{gupta2021gerrymandering,cygan2010exact}}
		Subsets $S_1, S_2 \subseteq U$ are disjoint if and only if Hamming weight of the string $\chi(S_1)+\chi(S_2)$ is $|S_1|+|S_2|$.
	\end{lemma}

	\begin{proof}[Proof of \Cref{thm:dk-2n}]
		We define the polynomials as follows. We first construct the Type-$1$ polynomial, in which for each $s\in [n]$, $\alpha \in [\bar{u}]$, and $\beta \in [\mathtt{b}]$, the non-zero polynomial $P^{1}_{s, \alpha, \beta} (x)$ denotes that there exists an $s$-sized subset of voters $Y\subseteq \X{V}$ corresponding to which there is an item $p\in \X{P}$ such that $\mathtt{c}(p) = \beta$ and $\mathtt{util}_Y(p)=\alpha$:
		$$P^{1}_{s, \alpha, \beta} (x) = \sum_{\substack{
				Y \subseteq \X{V}:~ |Y| = s \\
				\exists p \in \X{P}:~ \mathtt{util}_Y(p) = \alpha \\
				c(p) = \beta \le \mathtt{b}
		}} x^{\chi(Y)}.$$

		Next, for every $s\in [n]$, $\alpha \in [\bar{u}]$, $\beta \in [\mathtt{b}]$, and $j\in [2,n]$, we iteratively define the Type-$j$ polynomial as follows: $$P^{j}_{s, \alpha, \beta} (x) = \sum_{\substack{
				s_1, s_2 \in [n] :~ s_1 + s_2 = s \\
				\alpha_1, \alpha_2 \in [\bar{u}] :~ \alpha_1 + \alpha_2 = \alpha \\
				\beta_1, \beta_2 \in [\mathtt{b}]:~ \beta_1 + \beta_2 = \beta
		}} \mathcal{R}(\mathcal{H}_{s}(P^{1}_{s_1, \alpha_1, \beta_1} \times P^{j-1}_{s_2, \alpha_2, \beta_2})).$$

		A non-zero polynomial $P^{j}_{s, \alpha, \beta} (x)$ denotes that there exists $j$ disjoint voter subsets $Y_1,\ldots,Y_j$ such that $|Y_1|+\ldots+|Y_j|=s$, and there exists items $p_1,\ldots,p_j$ such that $\sum_{i\in [j]}\mathtt{util}_{Y_i}(p_i) = \alpha$ and $\sum_{i\in [j]}\mathtt{c}(p_i) = \beta \leq \mathtt{b}$. Among all polynomials, we check if for some $\alpha \ge \mathtt{u_{tgt}}$ and $\beta \leq \mathtt{b}$, the polynomial $P_{n, \alpha, \beta}^{k}$ is non-zero. If so, we return ``YES", otherwise we return ``NO". We prove the correctness of the algorithm in \Cref{proof-omit}.

		Note that the total number of polynomials is at most $n\bar{u}\mathtt{b}$.
		A polynomial has at most $2^n$ monomials, each of which has degree at most $2^n$. Given two polynomials of degree $2^n$, we can compute their product in time $\mathcal{O}(2^n n)$ \cite{moenck1976practical}. Hence, the running time is %
		$\mathcal{O}(2^n\ \text{poly}(m,n,\bar{u},\mathtt{b}))$.
	\end{proof}

	\noindent\textbf{FPTAS for \dk.}\label{approx} Due to the submodularity, \dk admits a factor $(1-\nicefrac{1}{e})$-approximation algorithm~\cite{fluschnik2019fair,sviridenko2004note}. Here, we give an \fptas when the profile is either SP or SC.

	The idea is to first scale down the utilities, round-off, and then use the known polynomial time algorithms for SP and SC profile under unary utilities~\cite{fluschnik2019fair}. %
	\begin{proposition}\label{prop:dk-sp}
		\label{prop:SPSCpolytime}
		\cite{fluschnik2019fair} \dk can be solved in time $\text{poly}(n,m,\hat{u})$, when the utility profile is encoded in unary, and is either SP or SC.
	\end{proposition}

	Let $u_{max}= \max_{v \in \X{V}, p \in \X{P}} \mathtt{util}_{v}(p)$, the maximum utility that a voter assigns to an item.  Let $0<\epsilon \le 1$ be the error parameter. We scale down the utility of every voter for every item by a factor of $s$, where $s=(\nicefrac{\epsilon}{2n})u_{max}$, as follows: $\overline{\mathtt{util}}_v(p)=\lceil \nicefrac{\mathtt{util}_v(p)}{s} \rceil$. Next, we round the utilities as follows: $\widetilde{\mathtt{util}}_v(p)=s \cdot \overline{\mathtt{util}}_v(p)$. Clearly,
	\begin{equation}\label{fptas:eq1}
	\mathtt{util}_v(p) \leq \widetilde{\mathtt{util}}_v(p) \leq \mathtt{util}_v(p)+s.
	\end{equation}
	For simplicity of analysis, we assume that $s$ is an integer. From now on, by $\mathtt{util}(S)$, we denote $\sum_{v \in \X{V}} \max_{p \in S} \mathtt{util}_v(p)$. Next, we show that the  total utility of the optimal bundle under the utilities $\widetilde{\mathtt{util}}$ is not very far from the total utility of the optimal bundle under the utilities ${\mathtt{util}}$. In particular, we prove the following lemma, where $\widetilde{\cal I} = (\X{P},\X{V},\mathtt{c},\{\widetilde{\mathtt{util}}_v\}_{v\in \X{V}}, \mathtt{b})$.
	\begin{lemma}[$\dagger$]\label{lem:fptas-rounding}
		Let $S^\star$ be an optimal solution to $\widetilde{\cal I}$. Let $S$ be any subset of $\X{P}$ such that $\sum_{p\in S}\mathtt{c}(p)\leq \mathtt{b}$. Then, $\mathtt{util}(S)\leq (1+\epsilon)~\mathtt{util}(S^\star)$.
	\end{lemma}

	In light of \Cref{lem:fptas-rounding}, our goal is reduced to finding an optimal solution for $\widetilde{\cal I}$. Note that for any $v\in \X{V}$ and $p\in \X{P}$, $\widetilde{\mathtt{util}_v(p)}$ and $\overline{\mathtt{util}}_v(p)$ differ by a factor of $s$. Thus, the utility of optimal solutions under these functions differ by a factor of $s$. Thus, our goal is reduced to finding an optimal solution for $\overline{\cal I} = (\X{P},\X{V},\mathtt{c},\{\overline{\mathtt{util}}_v\}_{v\in \X{V}},\mathtt{b})$. Note that by scaling the utilities, the profile still remains SP or SC. Thus, we can solve $\overline{\cal I}$ using \Cref{prop:dk-sp} in polynomial time. The running time is due to the fact that for any voter, the utility for any item under $\overline{\mathtt{util}}$ is at most $\nicefrac{2n}{\epsilon}$. Thus, we have the following theorem.
	\begin{theorem}\label{thm:dk-fptas}
		There exists an \fptas for \dk when the utility profile is either SP or SC.
	\end{theorem}

%% file: conclusion.tex
	\section{Conclusion}\label{conclusion}
	
	In this paper, we studied the computational and parameterized complexity of three multiagent variants of the knapsack problem. For \dk, we presented improved \fpt algorithms parameterized by $n$, an {\sf FPTAS} under SP and SC restrictions, and algorithms for special cases. Additionally, we gave some hardness results and algorithms for \muk and \buk.

	\dk is \nph even for SPPs, as shown in \cite{fluschnik2019fair}. Further, for unary encoded utilities, \dk can be solved in polynomial time for SPP and SCP. An open question is the complexity of \muk/\buk with SPP for $\lambda\geq 2$. In this paper, we studied two aggregation rules. One can study multiagent knapsack with other combinations of preference elicitation schemes and aggregation rules.

%% file: appendix.tex
	\section*{ \center\huge{Appendix}}
	
	\section{Details Omitted From The Main Text}
	
	\paragraph{\bf Algorithmic concepts.} A central notion in parameterized complexity is \emph{fixed-parameter tractability}. A parameterized problem $L \subseteq \Sigma^* \times \mathbb{N}$ is \emph{fixed-parameter tractable} (\fpt) with respect to the parameter $k$ (also denoted by {\sf{FPT}}($k$)), if for a given instance $(x,k)$, its membership in $L$ (i.e., $(x,k) \in L$) can be decided in time $f(k) \cdot {\sf poly}(|x|)$, where $f(\cdot)$ is an arbitrary computable function and ${\sf poly}(\cdot)$ is a polynomial function. However, all parameterized problems are not $\fpt$. Contrastingly, $\W$-hardness, captures the intractability in parameterized complexity. An \XP algorithm for $L$ with respect to $k$ can decide $(x,k) \in L$ in time $|x|^{f(k)}$, where $f(\cdot)$ is an arbitrary computable function.

	In some of our algorithms, we use the tool of a {\it reduction rule}, defined as a rule applied to the given instance of a problem to produce another instance of the same problem. A reduction rule is said to be {\it safe} if it is sound and complete, i.e., applying it to the given instance produces an equivalent instance. We refer the reader to books~\cite{downey2013fundamentals,cygan2015parameterized,DBLP:books/ox/Niedermeier06}.
	A central notion in the field of approximation algorithms is the {\it fully polynomial-time approximation scheme} (FPTAS). It is an algorithm that takes as input an instance of the problem and a parameter $\epsilon > 0$. It returns as output a solution whose value is at least $(1-\epsilon)$ (resp. $(1+\epsilon)$) times the optimal solution if it is a maximization (minimization) problem and runs in time polynomial in the input size and $1/\epsilon$.
	
	\section{Proofs Omitted From The Main Text} \label{proof-omit}
	
	\subsection{Details of \Cref{thm:im-sc-np-hardness}}
	
	\begin{claim}\label{clm:nph-sc-profile}
		The ordering $\sigma$ makes $\cal{I}'$ an SCP.
	\end{claim}
	\begin{proof}
		Consider two items $p,p' \in \X{P}$ such that $v_1$ prefers $p$ over $p'$. Since $\X{P}$ is the set of items in ${\cal I}$ and $\sigma$ is an SC ordering for ${\cal I}$, by the definition of $\mathtt{util}'$, it follows that there exists $t \in [n]$, such that the set of all the voters who prefer $p$ over $p'$ is $\{v_1,\ldots,v_t\}$. Consider an item $p\in \X{P}$ and an item $p' \in \X{P}_{new}$. Since, all the voters prefer $p'$ over $p$, we have $t = n$. Similarly, for $p,p' \in \X{P}_{new}$ such that $v_1$ prefers $p$ over $p'$, we have $t=n$.
	\end{proof}

	\begin{claim}\label{clm:nph-correctness}
		${\cal I}$ is a yes-instance of \dk if and only if ${\cal I}^\prime$ is a yes-instance of \mpb.
	\end{claim}
	\begin{proof}
		Let $S$ be a solution for ${\cal I}$. We claim that $S^\prime = S \cup \X{P}_{new}$ is a solution of ${\cal I}^\prime$. Note that any of the $\lambda-1$ items in $\X{P}_{new}$ cannot contribute to the satisfaction of any voter, i.e., the utility of these items do not contribute to $\sum_{v \in \X{V}} \mathtt{sat}^{\lambda}_{v}(S)$. Thus, for every voter $v$, the $\lambda^{\text{th}}$ most preferred item in $S'$ is the same as the most preferred item in $S$, a solution for ${\cal I}$. Therefore, $\sum_{v\in \X{V}}\mathtt{sat}^{\lambda}_{v}(S') \geq \mathtt{u_{tgt}}$.   %
		Since, the cost of picking $\X{P}_{new}$ is $\lambda - 1$, the cost of $S'$ is at most $\mathtt{b} + \lambda-1 = \mathtt{b}^\prime$.
		
		In the backward direction, let $S^\prime \subseteq \X{P}^\prime$ be a solution for ${\cal I}^\prime$. Since $\mathtt{u_{tgt}} > 0$, $S^\prime$ contains at least $\lambda$ items. Let $S_{cheap}$ be the set of $\lambda - 1$ cheapest items in $S^\prime$. Let $\bar{S} = (S^\prime \setminus S_{cheap}) \cup \X{P}_{new}$. Clearly, the cost of $\bar{S}$ is at most the cost of $S'$ since the cost of every item in $\X{P}_{new}$ is one, and the cost of every item in $S_{cheap}$ is at least one. By adding $\X{P}_{new}$, we have added $\lambda-1$ highest preferred items for all voters, in $\bar{S}$. Then, for each voter, the $\lambda^{th}$ most preferred item in $\bar{S}$ is from $S' \setminus S_{cheap}$. Therefore, the utility of $\bar{S}$ is at least as much as the utility of $S'$. We claim that $S = \bar{S} \setminus \X{P}_{new}$ is a solution to ${\cal I}$. Clearly, $S$ is a subset of $\X{P}$. Furthermore, the cost of $S$ is at most $\mathtt{b}^\prime - (\lambda-1) = \mathtt{b}$, and as argued in the forward direction, the total utility of $S$ is at least $\mathtt{u_{tgt}}$.
	\end{proof}

	\subsection{Details of \Cref{lem:fpt-wmsc}}
	
	\begin{claim}\label{clm:dp-wmsc}
		\Cref{eq1:dp-wmsc,eq2:dp-wmsc} compute $T[X,j]$ correctly for all $X\in {\cal R}$ and $j\in [0,m]$.
	\end{claim}
	\begin{proof}
		If $j=0$ and $\overrightarrow{X}=\{0\}^n$, then the optimal solution is the empty set of cost zero. If $j=0$ and $\overrightarrow{X} \ne \{0\}^n$, then $T[\overrightarrow{X}, 0] = \infty$, since no element of $U$ can be covered without picking any set. Thus, the base case is correct. Next, we prove the correctness of \Cref{eq2:dp-wmsc} by proving the inequalities in both directions. For a vector $\overrightarrow{X}\in {\cal R}$, if a set ${\cal F}^\prime \subseteq \{F_1,\ldots,F_j\}$ covers $u_i$ at least $\overrightarrow{X}_i$ times for each $i \in [n]$, then we say that ${\cal F}^\prime$ is a {\em valid} candidate for the entry $T[\overrightarrow{X},j]$. Let ${\cal F}^\prime$ be a minimum cost subset of $\{F_1,\ldots,F_j\}$ that covers $u_i$ at least $\overrightarrow{X}_i$ times, for each $i \in [n]$. If ${\cal F}^\prime$ contains $F_j$, then ${\cal F}^\prime \setminus \{F_j\}$ is a valid candidate for $T[\overrightarrow{X}-\chi(F_j),j-1]$. Thus, $T[\overrightarrow{X}-\chi(F_j),j-1]\leq \sum_{F\in {\cal F}^\prime}\mathtt{c}(F)-\mathtt{c}(F_j)$.  Else, ${\cal F}^\prime$ is a valid candidate for $T[\overrightarrow{X},j-1]$. Thus, $T[\overrightarrow{X},j-1]\leq \sum_{F\in {\cal F}^\prime}\mathtt{c}(F)$. Hence,
		\begin{equation*}
		T[\overrightarrow{X},j]\geq \min\{T[\overrightarrow{X},j-1],T[\overrightarrow{X}-\chi(F_j),j-1]+\mathtt{c}(F_j)\}.
		\end{equation*}
		Similarly, in the other direction, let ${\cal F}^\prime$ be a minimum cost subset of $\{F_1,\ldots,F_{j-1}\}$ that covers $u_i$ at least $\overrightarrow{X}_i$ times, for each $i \in [n]$. Then, ${\cal F}^\prime$ is a valid candidate for $T[\overrightarrow{X},j]$. Thus, $T[\overrightarrow{X},j] \leq T[\overrightarrow{X},j-1]$. Furthermore, for any valid candidate ${\cal F}'$ of $T[\overrightarrow{X}-\chi(F_j),j-1]$, ${\cal F}^\prime \cup \{F_j\}$ is also a valid candidate for $T[\overrightarrow{X},j]$. Hence,
		\begin{equation*}
		T[\overrightarrow{X},j]\leq \min\{T[\overrightarrow{X},j-1],T[\overrightarrow{X}-\chi(F_j),j-1]+\mathtt{c}(F_j)\}.
		\end{equation*}
		Hence, the equality holds, and the claim follows from mathematical induction.
	\end{proof}
		
	\subsection{Correctness of the Algorithm in \Cref{thm:muk-fpt-nb}}

	\begin{lemma}
		$\cal{I}$ is a yes-instance of $\mpb$ if and only if the algorithm returns ``YES".
	\end{lemma}
	\begin{proof}
		Let $S'=\{s_1', \dots, s_k'\}$ be the bundle corresponding to $M$. Let $k$, $\{\psi_i : v_i \in \X{V}\}$ be the guesses corresponding to $S'$. By the knapsack constraint, $\mathtt{c}(S') \le \mathtt{b}$. Since the weight of the matching is at least $\mathtt{u_{tgt}}$, we have that $\sum_{v_i \in \X{V}} \mathtt{sat}^\lambda_{v_i} (S') \ge \sum_{v_i \in \X{V}} \sum_{s_j' \in S'} \mathbb{I}[\psi_i(j)] \cdot \mathtt{util}_{v_i}(s_j') \ge \mathtt{u_{tgt}}$. For the other direction, let $A = \{a_1, \dots, a_k\}$ be a solution bundle. For each voter $v_i$, we can get a correct guess of permutation $\psi_i$ by projecting $A$ on $\succ_i$. Let $B = \{b_1, \dots, b_r\}$ be the bundle obtained by our algorithm. Let $r$, $\{\psi_i' : v_i \in \X{V}\}$ be the guesses corresponding to $B$. Then, $\sum_{v_i \in \X{V}} \mathtt{sat}^\lambda_{v_i} (B) \ge \sum_{v_i \in \X{V}} \sum_{b_j \in B} \mathbb{I}[\psi_i'(j)] \cdot \mathtt{util}_{v_i}(b_j) \ge  \sum_{v_i \in \X{V}} \sum_{a_j \in A} \mathbb{I}[\psi_i(j)] \cdot \mathtt{util}_{v_i}(a_j) \ge \mathtt{u_{tgt}}$. Note that the second inequality holds because we pick the bundle corresponding to the matching with maximum weight among all guesses.
	\end{proof}

	\subsection{Safeness of \Cref{rr:muk-su}}
	
	\begin{lemma}\label{lem:rr-safeness}
		${\cal I}$ is a yes-instance of \muk/\buk if and only if ${\cal I}'$ is a yes-instance of \muk/\buk.
	\end{lemma}
	\begin{proof}
	Since the preference orders of $v_1$ and $v_2$ are identical, for any bundle, they are represented by the same item.

	Let $S$ be a feasible bundle for ${\cal I}$. Let $p \in S$ be the representative of $v_1$ and $v_2$ in $S$. We claim that $S$ forms a feasible bundle for ${\cal I'}$. The cost of $S$ is at most $\mathtt{b}$, and the utility of $S$ is at least $\mathtt{u_{tgt}}$ because $\sum_{v \in \X{V}^\prime} \mathtt{sat}_{v}^{\lambda}(S) = \sum_{v \in \X{V} \setminus \{v_1, v_2\}} \mathtt{sat}_{v}^{\lambda}(S) + \mathtt{util}_{v_{12}}(p) = \sum_{v \in \X{V}} \mathtt{sat}_{v}^{\lambda}(S) \ge \mathtt{u_{tgt}}.$

	For the other direction, let $S'$ be a feasible bundle for ${\cal I'}$. Let $p' \in S'$ be the representative of $v_{12}$ in $S'$. Similar to the other direction, we claim that $S'$ forms a feasible bundle for ${\cal I}$. The cost of $S'$ is at most $\mathtt{b}$, and the utility of $S'$ is at least $\mathtt{u_{tgt}}$ because $\sum_{v \in \X{V}} \mathtt{sat}_{v}^{\lambda}(S') = \sum_{v \in \X{V}^\prime \setminus \{v_{12}\}} \mathtt{sat}_{v}^{\lambda}(S') + \mathtt{util}_{v_1}(p') + \mathtt{util}_{v_2}(p') = \sum_{v \in \X{V}^\prime} \mathtt{sat}_{v}^{\lambda}(S') \ge \mathtt{u_{tgt}}$.		
	\end{proof}

	\subsection{Proof of \Cref{lem:muk-strongly-unan}}
	The claim on \muk follows from \Cref{thm:muk-xp}. For the claim on \buk, we present an algorithm to solve the problem when the input has exactly one voter. 
	
	Using the technique of dynamic programming, we fill a table $T$, defined as follows: for every $i \in [0,m], j \in [0,\T{b}], k \in [0,\lambda]$, let $T[i,j,k]$ be the maximum value of a bundle that costs at most $j$, and has exactly $k$ items, all of which are chosen from the first $i$ items - $\{a_1, \dots, a_i\}$. We call such a bundle a \textit{valid candidate} for the entry $T[i,j,k]$. We compute the table entries as follows.
	
	\paragraph{Base Case:} For any $i \in [0,m], j \in [0,\T{b}], k \in [0,\lambda]$, we set 
		\begin{equation}\label{eq:buk-su-base-case}
		T[i,j,k] =
		\begin{cases}
		-\infty & \quad\text{if } i < k \\
		0      & \quad\text{otherwise.}\\
		\end{cases}
		\end{equation}
	
	\paragraph{Recursive Step:} For any $i \in [m], j \in [\T{b}], k \in [\lambda]$, we set 
	\begin{equation}\label{eq:buk-su-rec-step}
	T[i,j,k] = \min\{T[i-1,j-\T{c}(a_i), k-1]+\T{util}_v(a_i), T[i-1, j, k]\}.
	\end{equation}
	
	The algorithm returns ``YES" if $T[n,j,k] \ge \T{u_{tgt}}$ for some $j \in [\T{b}]$ and $k \in [\lambda]$, ``NO" otherwise. Next, we prove the algorithm's correctness through the following claim.
	
	\begin{claim}\label{clm:buk-su-dp}
		\Cref{eq:buk-su-base-case,eq:buk-su-rec-step} compute $T[i,j,k]$ correctly, for every $i \in [0,n], j \in [0,\T{b}]$, and $k \in [0,\lambda]$.
	\end{claim}
	\begin{proof}
		If $k>i$, then for any $j \in [\T{b}]$, $T[i,j,k]=-\infty$ since there is no bundle that contains $k$ out of the first $i$ items. Thus, \Cref{eq:buk-su-base-case} is correct.

		Next, we prove the correctness of \Cref{eq:buk-su-rec-step} by showing inequalities in both directions. In one direction, let $S$ be a valid candidate for the entry $T[i,j,k]$. Now, we consider the two possible cases. If $a_i \in S$, then $S \setminus \{a_i\}$ is a valid candidate for the entry $T[i-1, j-\T{c}(a_i), k-1]$. Otherwise, we have that $a_i \notin S$. Then, $S$ is a valid candidate for the entry $T[i-1, j, k]$. Hence, $$T[i,j,k] \le \min\{T[i-1,j-\T{c}(a_i), k-1]+\T{util}_v(a_i), T[i-1, j, k]\}.$$

		In the other direction, any valid candidate for the entry $T[i-1,j,k]$ is also a valid candidate for the entry $T[i,j,k]$. Also, for any valid candidate $S$ for the entry $T[i-1, j-\T{c}(a_i), k-1]$, $S \cup \{a_i\}$ is a valid candidate for the entry $T[i,j,k]$. Hence, $$T[i,j,k] \ge \min\{T[i-1,j-\T{c}(a_i), k-1]+\T{util}_v(a_i), T[i-1, j, k]\}.$$
		
		Thus, \Cref{eq:buk-su-rec-step} is correct, and the claim follows from mathematical induction.
	\end{proof}
	Since the table size is $\OO(n \T{b} \lambda)$ and each table entry can be computed in polynomial time, the running time in the lemma follows. An alternative dynamic-programming approach to solve the problem is that instead of finding the maximum utility of a $k$-sized bundle among the first $i$ items with budget $j$, we find the minimum cost of a $k$-sized bundle among the first $i$ items with utility $j$. This results in a table of size $\OO(n \bar{u} \lambda)$. Thus, we have the last part of the lemma.
	\qed

	\subsection{Proof of \Cref{lem:dk-unanimous}}
	
	Since the profile is a UP, there exists an item $p \in \X{P}$ that is the most preferred by every voter. We know that $\mathtt{c}(p) \le \mathtt{b}$. Therefore, it forms a feasible bundle. We claim that $\{p\}$ forms an optimal bundle since $p$ is every voter's representative as $\lambda=1$. \qed

	\subsection{Details of \Cref{lem:dk-kpcover}}
	
	\begin{claim}\label{clm:dk-kpcover}
		${\cal I}$ is a yes-instance of \dk if and only if ${\cal J}$ is a yes-instance of \kpcover.
	\end{claim}
	\begin{proof}
		In the forward direction, let $S$ be a solution to ${\cal I}$. For every $r\in S$, let $V_r\subseteq \X{V}$ be the set of voters who are represented by the item $r$, i.e.,  $V_r = \{v \in \X{V} \colon r = \argmax_{q \in S} \mathtt{util}_v(q)\}$. Without loss of generality, we assume that for every pair of items $\{r,r'\}\subseteq S$, $V_r \cap V_{r'} = \emptyset$; if two or more items in $S$ can represent a voter since they have equal utilities, then we can arbitrarily pick one of those items $r$ and place the voter in $V_r$. Let $Z=\{\{V_r,r\} \colon r\in S\}$. Note that $Z\subseteq \X{F}$. Since every voter is represented by some item in $S$, every element of $U_1$ (voter) belongs to some $V_r$, where $r\in S$ (the item that represents the voter). Since $\mathtt{b}'=\mathtt{b}$, $\sum_{\{V_r,r\}\in Z}\mathtt{cost}(\{V_r,r\})\leq \mathtt{b}'$. Furthermore, note that for every $r\in S$, $\mathtt{profit}(\{V_r,r\})= \sum_{v\in V_r}\mathtt{sat}_v^1(S)$. Therefore, $\sum_{\{V_r,r\}\in Z}\mathtt{profit}(\{V_r,r\})=   \sum_{v\in \X{V}} \mathtt{sat}_v^1(S) \geq \mathtt{u_{tgt}}$. Since $\mathtt{p} = \mathtt{u_{tgt}}$, $Z$ is a solution to ${\cal J}$.
		
		In the backward direction, let $Z$ be a solution to ${\cal J}$. Let $S=\{r\in \X{P}\colon \{X,r\}\in Z\}$. Since $\mathtt{b}'=\mathtt{b}$, $\sum_{r\in S}\mathtt{c}(r)\leq \mathtt{b}$. Note that for every $\{X,r\}\in Z$, the total satisfaction of all the voters $v\in X$, from $S$, i.e., $\sum_{v\in X}\mathtt{sat}_v^1(S)$ is at least $\mathtt{profit}(\{X,r\})$ as $r\in S$. Since $\bigcup_{\{X,r\}\in Z}X = U_1$, we have that $\sum_{v\in \X{V}}\mathtt{sat}_v^1(S) \geq \sum_{\{X,r\}\in Z}\mathtt{profit}(\{X,r\})$. Thus, $\sum_{v\in \X{V}} \mathtt{sat}_v^1(S) \ge \mathtt{u_{tgt}}$, as $\mathtt{u_{tgt}}=\mathtt{p}$, and \\$\sum_{\{X,p\}\in Z}\mathtt{profit}(\{X,p\}) \geq \mathtt{p}$.
	\end{proof}

	\subsection{Proof of \Cref{thm_dk-kpcover-complexity}}
	
	The algorithm uses the technique of dynamic programming, and it is similar to the algorithm for {\sc Set Cover} given by Fomin et al.~\cite{fomin2004exact}. Let ${\cal J}=(U_1,U_2,\X{F},\mathtt{cost},\mathtt{profit}, \mathtt{b}, \mathtt{p})$ be an instance of \kpcover. Let $\X{F}=\{F_1,\ldots,F_{m}\}$. We define the dynamic-programming table as follows: for every subset $S\subseteq U_1$, $j\in [p_{max}]$, and $k\in [m]$, we define $T[S,k,j]$ as the minimum cost of a subset $\X{F'}\subseteq \{F_1, \ldots, F_k\}$ with profit exactly $j$ such that $\bigcup_{F\in \X{F'}}F\setminus U_2 = S$, and for every pair of sets $\{F,F'\}\subseteq \X{F'}$, $(F\cap F')\setminus U_2 = \emptyset$. We call such a set $\X{F'}$ as a {\em valid candidate} for the entry $T[S,k,j]$. If no such subset exists, then $T[S,k,j]=\infty$. We compute the table entries as follows.

	\paragraph{Base Case:} For each $S \subseteq U_1$ and $j \in [0, p_{max}]$, we set
	\begin{equation}\label{eq:kpcover-base-case1}
	T[S,0,j] =
	\begin{cases}
	0      & \quad\text{if } S = \phi, j=0 \\
	\infty & \quad\text{otherwise.}        \\
	\end{cases}
	\end{equation}

	\paragraph{Recursive Step:} For every subset $S\subseteq U_1$, $k\in [m]$, and $j\in [0,p_{max}]$, we set
	\begin{equation}\label{eq:kpcover-recursion}
	\begin{split}
	T[S,k,j] = & \min\{T[S, k-1, j], T[S\setminus F_k, k-1, \\
	& j-\mathtt{profit}(F_k)] + \mathtt{cost}(F_k)\}
	\end{split}
	\end{equation}

	The algorithm returns ``YES" if $T[U_1, m, j] \le \mathtt{b}$ for some $j \ge \mathtt{p}$, ``NO" otherwise. Algorithm~\ref{algo:kpcover} presents the pseudocode of this algorithm.

	\begin{algorithm2e}
		\KwIn{an instance $(U_1,U_2,\X{F},\mathtt{cost},\mathtt{profit}, \mathtt{b}, \mathtt{p})$ of \kpcover}
		\KwOut{``YES" if there exists a solution to $(U_1,U_2,\X{F},\mathtt{cost},\mathtt{profit}, \mathtt{b}, \mathtt{p})$, ``NO" otherwise.}
		\caption{An Algorithm for \kpcover}\label{algo:kpcover}
		\For{{\rm\bf each} $S \subseteq U_1$, $i\in [0,m]$, and $j \in [0, p_{max}]$}{
			$T[S,i,j]=0$
		}
		\For{{\rm\bf each} $S \subseteq U_1$ and $j \in [p_{max}]$, set} {$T[S,0,j] =
			\infty \quad \text{ if } S \neq \phi$}
		
		\For{$S \subseteq U_1$, $i\in [m]$, $j \in [0, p_{max}]$}{
			$T \left[S, i, j \right]$ $=$ $\min \left(
			\begin{array}{l}
			T \left[S, i-1, j \right],                                     \\
			T \left[S \backslash F_i, i-1, j - \mathtt{profit}(F_i)\right] \\+ \mathtt{cost}(F_i)
			\end{array}
			\right)$
		}
		Return ``YES" if $T \left[U_1, m, j\right] \le \mathtt{b}$ for some $j \ge \mathtt{p}$, ``NO" otherwise.
	\end{algorithm2e}

	Next, we prove the correctness of the algorithm. In particular, we prove the following claim.

	\begin{claim}
		\label{clm_dk-kpcover-algo-correctness}
		\Cref{eq:kpcover-base-case1,eq:kpcover-recursion} compute $T[S,k,j]$ correctly, for every $S\subseteq U_1$, $k\in [0,m]$, and $j\in [0,p_{max}]$.
	\end{claim}
	\begin{proof}
		If $j=0$ and $S=\phi$, then the optimal solution is the empty set of zero cost. If $j>0$ and $k=0$, then there is no solution ($T[S,0,j] = \infty$) since we cannot obtain non-zero profit $j$ without picking any set. If $S \ne \phi$ and $j=0$, then there is no solution ($T[S,k,0] = \infty$) since we cannot cover a non-empty subset of $U_1$ without picking any set. This proves the correctness of \Cref{eq:kpcover-base-case1}. %
		
		Next, we prove the correctness of \Cref{eq:kpcover-recursion} by showing inequalities in both directions. In one direction, let $\X{F}' \subseteq \{F_1,\ldots,F_i\}$ be a minimum cost valid candidate for the entry $T[S,i,j]$. Now, there are two cases: either $F_i \in \X{F}'$ or not. In the latter case, $\X{F}'$ is also a valid candidate for the entry $T[S,i-1,j]$ as $\X{F}'\subseteq \{F_1,\ldots,F_{i-1}\}$. Therefore, $T[S,i-1,j] \leq T[S,i,j]$. In the former case, consider the set $\X{F}''=\X{F}'\setminus \{F_i\}$. Note that the profit of the set $\X{F}''$ is $\sum_{F\in \X{F}'}\mathtt{profit}(F)-\mathtt{profit}(F_i)\geq j-\mathtt{profit}(F_i)$. Thus, $\X{F}''$ is a valid candidate for the entry $T[S\setminus F_i,i-1,j-\mathtt{profit}(F_i)]$. Therefore, $T[S\setminus F_i,i-1,j-\mathtt{profit}(F_i)] \leq T[S,i,j] - \mathtt{cost}(F_i)$. Hence,
		\begin{equation*}
		\begin{split}
		T[S,i,j] \geq & \min\{T[S, i-1, j], \\
		& T[S\setminus F_i, i-1,
		j-\mathtt{profit}(F_i)] + \mathtt{cost}(F_i)\}
		\end{split}
		\end{equation*}

		For other direction, let $\X{F}'\subseteq \{F_1,\ldots,F_{i-1}\}$ be a valid candidate for the entry $T[S,i-1,j]$. Then, clearly, it is also a valid candidate for the entry $T[S,i,j]$. Therefore, $T[S,i,j]\leq T[S,i-1,j]$. Similarly, if $\X{F}'$ is a valid candidate for $T[S\setminus F_i,i-1,j-\mathtt{profit}(F_i)]$, then $\X{F'}\cup \{F_i\}$ is a valid candidate for the entry $T[S,i,j]$. Hence,
		\begin{equation*}
		\begin{split}
		T[S,i,j] \leq & \min\{T[S, i-1, j], \\
		& T[S\setminus F_i, i-1,
		j-\mathtt{profit}(F_i)] + \mathtt{cost}(F_i)\}
		\end{split}
		\end{equation*}

		This proves the equality. Since the recursive formula is correct, the claim follows from mathematical induction.
	\end{proof}
	Since the table size is $\OO(2^{|U_1|}|{\cal F}|p_{max})$ and each entry can be computed in polynomial time, the running time of the algorithm follows.
	\qed

	\subsection{Correctness of the Algorithm in \Cref{thm:dk-2n}}
	\begin{lemma}\label{lem:correctness-1}
		If ${\cal I}$ is a yes-instance of \dk, the algorithm returns ``YES". %
	\end{lemma}
	\begin{proof}
		Let $S = \{p_1, \ldots, p_k\}$ be a cost-minimal solution to ${\cal I}$. %
		For every $p_\ell \in S$, let $Y_\ell \subseteq \X{V}$ be the set of voters who are represented by $p_\ell$, i.e.,  $Y_\ell = \{v \in \X{V} \colon p_\ell = \argmax_{q \in S} \mathtt{util}_v(q)\}$. Note that for any $\ell \in [k]$, $Y_\ell \neq \phi$, because otherwise we can simply remove $p_\ell$ from $S$ since it does not contribute to the utility of $S$.
		Without loss of generality, we assume that for every pair of items $\{p_{i},p_{j}\}\subseteq S$, $Y_i \cap Y_j = \emptyset$; if two or more items in $S$ can represent a voter since they have equal utilities, then we can arbitrarily pick one of those items $p_\ell$, and place the voter in $Y_\ell$.

		\begin{claim}\label{clm:dk-algo-polynomial}
			For each $t \in [k]$, the polynomial $P^{t}_{s, \alpha, \beta}$ is non-zero, where $s = \sum_{\ell \in [t]} \lvert Y_\ell \rvert$, $\beta = \sum_{\ell \in [t]} \mathtt{c}(p_\ell)$, and $\alpha = \sum_{\ell \in [t]} \mathtt{util}_{Y_\ell}(p_\ell)$. In fact, it contains the monomial $x^{\chi (Y_1 \cup \ldots \cup Y_t)}$.
		\end{claim}
		\begin{proof}
			We will prove by induction on $t$.
			For $t=1$, by the construction of Type-$1$ polynomials, for $\alpha = \mathtt{util}_{Y_1}(p_1)$ and $\beta = \mathtt{c}(p_1)$, $P^1_{|Y_1|, \alpha, \beta}$ contains the monomial $x^{\chi(Y_1)}$.
			Let the claim be true for $t\leq t'$.
			We prove it for $t=t'+1$. Due to the induction hypothesis, we know that  for $s = \sum_{\ell \in [t']} | Y_\ell |$, $\alpha= \sum_{\ell \in [t']} \mathtt{util}_{Y_\ell}(p_\ell)$, and $\beta = \sum_{\ell \in [t']} \mathtt{c}(p_\ell)$, $P^{t'}_{s, \alpha, \beta} (x)$ contains the monomial $x^{\chi(Y_1 \cup \ldots \cup Y_{t'})}$.  Furthermore, due to the construction of Type-$1$ polynomials, we know that for $s'=|Y_{t+1}|$, $\alpha'=\mathtt{util}_{Y_{t+1}}(p_{t+1})$, and $\beta'=\mathtt{c}(p_{t+1})$,
			$P^{1}_{s', \alpha',\beta'} (x)$ contains the monomial $x^{\chi(Y_{t+1})}$. Since for all $\ell \ne \ell' \in [k]$, $Y_\ell \cap Y_{\ell'} = \emptyset$, we have that $Y_{t+1}$ is disjoint from $Y_1 \cup \ldots \cup Y_t$. Thus, $P^{t+1}_{s+s', \alpha+\alpha', \beta+\beta'}(x)$ contains the monomial $x^{\chi(Y_1 \cup \ldots \cup Y_{t+1})}$.
		\end{proof}

		Due to the above claim, for $\alpha=\mathtt{util}_{\X{V}}(S)$ and $\beta = \sum_{p\in S}\mathtt{c}(p)$, $P^{k}_{n, \alpha,\beta}(x)$ is non-zero. Since $S$ is a solution to ${\cal I}$, $\sum_{v\in \X{V}}\mathtt{sat}_v^1(S) \geq \mathtt{u_{tgt}}$ and  $\sum_{p\in S}\mathtt{c}(p) \leq \mathtt{b}$, Hence, the algorithm returns ``YES". \qed
	\end{proof}

	\begin{lemma}\label{lem:correctness-2}
		Suppose that there is a non-zero polynomial $P^{k}_{n, \alpha, \beta}$, for some $\alpha \ge \mathtt{u_{tgt}}, \beta \le \mathtt{b}$. Then, ${\cal I}$ is a yes-instance.  
	\end{lemma}
	\begin{proof}		
		We show that for any type $j \in [k]$, if $P^{j}_{s, \alpha, \beta}(x)$ contains a monomial $x^S$, then there are $j$ pairwise disjoint subsets $Y_1, \ldots, Y_j$ of $\X{V}$ such that $|Y_1\cup \ldots \cup Y_j|=s$ and $\chi(Y_1\cup \ldots \cup Y_j)=S$. Additionally, there exists a $j$-sized subset $\{ p_1, \ldots, p_j \}$ with the following properties: \begin{enumerate}
			\item $\alpha = \sum_{l \in [j]} \mathtt{util}_{Y_l}(p_l)$
			\item $\beta =  \sum_{l \in [j]} \mathtt{c}(p_l)$
		\end{enumerate}
		This will prove that ${\cal I}$ is a yes-instance.

		We prove it by induction on $j$. For $j=1$, suppose that for some $s \in [n], \alpha\in [\bar{u}], \beta\in [\mathtt{b}]$, polynomial $P^1_{s, \alpha, \beta}(x)$ is non-zero. Then, due to the construction of Type-I polynomials, it has a monomial $x^{\chi(Y_1)}$, where $|Y_1|=s$. Furthermore, there exists an item $p_1 \in \X{P}$ such that $\mathtt{c}(p_1) = \beta$ and $\mathtt{util}_{Y_1}(p_1) = \alpha$. We will assume that the statement is true for some $t=t'$ and prove for $t=t'+1$.

		Suppose that $P^{t'+1}_{s, \alpha, \beta}(x)$ contains a monomial $y^S$, for some $s\in [n], \alpha \in [\bar{u}], \beta \in [\mathtt{b}]$. Due to the construction of Type-$(t'+1)$ polynomial, there exists $s_1,s_2 \in[n]$ such that $s_1+s_2=s$, $\alpha_1,\alpha_2\in [\bar{u}]$ such that $\alpha_1+\alpha_2=\alpha$, and $\beta_1,\beta_2 \in [\mathtt{b}]$ such that $\beta_1+\beta_2 = \beta$, such that the polynomial $P^{t'+1}_{s, \alpha, \beta}(x)$ is obtained by the multiplication of $P^{t'}_{s_1,\alpha_1,\beta_1}$ and $P^1_{s_2,\alpha_2,\beta_2}$. Let $y^S=y^{S'}\times y^{S''}$, where $y^{S'}$ is a monomial in $P^{t'}_{s_1,\alpha_1,\beta_1}$ and $y^{S''}$ is a monomial in $P^1_{s_2,\alpha_2,\beta_2}$. Due to  the induction hypothesis, there exists pairwise disjoint subsets $Y_1,\ldots,Y_{t'}$ of $\X{V}$ such that $\chi(Y_1\cup\ldots \cup Y_{t'})=S'$ and there exists a subset $Y \subseteq \X{V}$ such that $\chi(Y)=S''$. Since  ${\cal H}(S' + S'') = {\cal H}(S')+{\cal H}(S'')$, due to \Cref{lem:distinct-vectors}, we know that $S'$ and $S''$ are disjoint. Thus, $Y$ is disjoint from every set in  $Y_1,\ldots,Y_{t'}$. Furthermore, due to induction hypothesis, there are items $p_1,\ldots,p_{t'}$ such that $\sum_{\ell \in [t']}\mathtt{util}_{Y_\ell}(p_\ell) = \alpha_1$, $\sum_{\ell \in [t']}\mathtt{c}(p_\ell)= \beta_1$, and there exists an item $p$ such that $\mathtt{util}_{Y}(p)=\alpha_2$ and $\mathtt{c}(p)= \beta_2$. Since every two sets in $Y_1,\ldots,Y_{t'},Y$ are pairwise disjoint, we have that $\sum_{\ell \in [t']}\mathtt{util}_{Y_\ell}(p_\ell)+\mathtt{util}_{Y}(p) = \alpha$ and $\sum_{\ell \in [t']}\mathtt{c}(p_\ell)+\mathtt{c}(p)= \beta$.
	\end{proof}

	\subsection{Proof of \Cref{lem:fptas-rounding}}
	
	Since $S^\star$ is an optimal solution to $\widetilde{\cal I}$, we know that
	\begin{equation}\label{fptas:eq2}
	\widetilde{\mathtt{util}}(S) \leq \widetilde{\mathtt{util}}(S^\star)
	\end{equation}

	Due to \Cref{fptas:eq1,fptas:eq2}, we have the following:
	\begin{equation}\label{fptas:eq3}
	\begin{split}
	{\mathtt{util}}(S) & \leq \widetilde{\mathtt{util}}(S) \leq \widetilde{\mathtt{util}}(S^\star) \\
	& \leq \sum_{v\in \X{V}}\max_{p\in S^\star} \left({\mathtt{util}}_v(p) + s \right) \\
	& \leq ns + \mathtt{util}(S^\star)
	\end{split}
	\end{equation}

	Consider $v' \in \X{V}, p' \in \X{P}$ such that $\mathtt{util}_{v'}(p')=u_{max}$. By our choice of $s$, $\mathtt{util}_{v'}(p')=(\nicefrac{2ns}{\epsilon})$. Hence, $\widetilde{\mathtt{util}}_{v'}(p')=\mathtt{util}_{v'}(p')$. Since $S^\star$ is an optimal solution to $\widetilde{\cal I}$, we know that
	\begin{equation}\label{fptas:eq4}
	\widetilde{\mathtt{util}}(S^\star) \geq \mathtt{util}(S^\star) \geq \mathtt{util}_{v'}(p') = \frac{2ns}{\epsilon}
	\end{equation}

	Due to \Cref{fptas:eq3}, we have that
	\begin{equation}\label{fptas:eq5}
	{\mathtt{util}}(S^\star) \geq \widetilde{\mathtt{util}}(S^\star) -ns
	\end{equation}

	Due to \Cref{fptas:eq4,fptas:eq5}, we have that
	\begin{equation}
	{\mathtt{util}}(S^\star) \geq \Big(\frac2{\epsilon}-1\Big)ns
	\end{equation}

	Since $0 < \epsilon \leq 1$, we have that $\left( \nicefrac2\epsilon -1 \right)^{-1} \le \epsilon$. Thus, $ns \leq \epsilon ~\mathtt{util}(S^\star)$, and the claim follows due to \Cref{fptas:eq3}.
	\qed

\section{More Related Work On PB}
	In PB, the city residents are asked for their opinion on the projects to be funded for the city, and then the preferences of all the voters are aggregated using some voting rule, which is used to decide the projects for the city. Initiated in Brazil in 1989 in the municipality of Porto Alegre~\cite{brazil}, PB has become quite popular worldwide~\cite{pbglobe}, including in the United States~\cite{usapb} and Europe~\cite{pbeurope}. In the last few years, PB has gained considerable attention from computer scientists~\cite{DBLP:conf/atal/0001LT18,DBLP:conf/aaai/TalmonF19,fluschnik2019fair,DBLP:conf/aaai/0001L21,aziz2021participatory,DBLP:conf/atal/0001TB21,DBLP:conf/ijcai/0001STZ21,DBLP:conf/ijcai/0001ST20,DBLP:conf/aaai/HershkowitzKPP21,pierczynski2021proportional}.  We would like to note that Goel et. al~\cite{Goel19} introduced the topic of {\it knapsack voting} that captures the process of aggregating the preferences of voters in the context of PB. The authors state that their motivation was to incorporate the classical knapsack problem by making the voter choose projects under the budget constraint but in a manner that aligns the constraints on the voters’ decisions with those of the decision-makers. Their study is centered around strategic issues and extends knapsack voting further to more general settings with revenues, deficits, and surpluses.